\documentclass[12pt]{article}

\usepackage{amsfonts}
\usepackage{amsmath}
\usepackage{amssymb}
\usepackage{amsthm}
\usepackage{mathtools}
\usepackage{authblk}
\usepackage{appendix}
\usepackage{bm}
\usepackage{bbm}
\usepackage{dsfont}
\usepackage{gensymb}
\usepackage{natbib}
\usepackage{subcaption}
\usepackage{lineno}
\usepackage{enumitem}
\usepackage[T1]{fontenc}
\usepackage{multirow}
\usepackage{soul}

\newcommand*{\textcal}[1]{%
  \textit{\fontfamily{qzc}\selectfont#1}%
}

\usepackage{graphicx,url}

\newtheorem{theorem}{Theorem}[section]

\newtheorem{remark}{Remark}
\newtheorem{definition}{Definition}

\usepackage[utf8]{inputenc}  

\usepackage[dvipsnames]{xcolor}
\newcommand{\bA}{\textbf{A}}
\newcommand{\bB}{\textbf{B}}
\newcommand{\bD}{\textbf{D}}
\newcommand{\bH}{\textbf{H}}
\newcommand{\bI}{\textbf{I}}
\newcommand{\bF}{\textbf{F}}
\newcommand{\bK}{\textbf{K}}
\newcommand{\bL}{\textbf{L}}
\newcommand{\bQ}{\textbf{Q}}
\newcommand{\bS}{\textbf{S}}

\newcommand{\bX}{\textbf{X}}
\newcommand{\bY}{\textbf{Y}}
\newcommand{\bZ}{\textbf{Z}}

\newcommand{\bd}{\textbf{d}}
\newcommand{\bh}{\textbf{h}}
\newcommand{\bq}{\textbf{q}}

\newcommand{\bz}{\textbf{z}}

\newcommand{\E}{\mathbb{E}}
\newcommand{\var}{\text{var}}
\newcommand{\cov}{\text{cov}}
\newcommand{\T}{\intercal}
\newcommand{\zero}{\bm{0}}
\newcommand{\one}{\bm{1}}

\newcommand{\blambda}{\bm{\lambda}}

\newcommand{\norm}[1]{\left\lVert#1\right\rVert}

\graphicspath{ {../images/} }

\newcommand{\blind}{1}
\newcommand{\spacing}{1.1}

\begin{document}

\def\spacingset#1{\renewcommand{\baselinestretch}%
{#1}\small\normalsize} \spacingset{1}

\if1\blind
{
  \title{\bf Generalised Bayes Linear Inference}
    \author[1]{Lachlan Astfalck}
    \author[2]{Cassandra Bird}
    \author[3]{Daniel Williamson}
    \affil[1]{School of Physics, Mathematics and Computing, The University of Western Australia, Crawley, Australia}
    \affil[2]{Department of Mathematics and Statistics, University of Exeter, Exeter, UK}
    \affil[3]{Land Environment Economics and Policy Institute, Department of Economics, University of Exeter, Exeter, UK}
    
    \setcounter{Maxaffil}{0}
    \renewcommand\Affilfont{\itshape\small}
  \maketitle
} \fi

\if0\blind
{
  \bigskip
  \bigskip
  \bigskip
  \begin{center}
    {\LARGE\bf Generalised Bayes Linear Inference}
\end{center}
  \medskip
} \fi

\bigskip
\begin{abstract}
Motivated by big data and the vast parameter spaces in modern machine learning models, optimisation approaches to Bayesian inference have seen a surge in popularity in recent years. In this paper, we address the connection between the popular new methods termed generalised Bayesian inference and Bayes linear methods. We propose a further generalisation to Bayesian inference that unifies these and other recent approaches by considering the Bayesian inference problem as one of finding the closest point in a particular solution space to a data generating process, where these notions differ depending on user-specified geometries and foundational belief systems. Motivated by this framework, we propose a generalisation to Bayes linear approaches that enables fast and principled inferences that obey the coherence requirements implied by domain restrictions on random quantities. We demonstrate the efficacy of generalised Bayes linear inference on a number of examples, including monotonic regression and inference for spatial counts. This paper is accompanied by an \texttt{R} package available at \url{github.com/astfalckl/bayeslinear}.
\end{abstract}

\noindent%
\vfill

\newpage
\spacingset{\spacing}
% \linenumbers
% \modulolinenumbers[1]

\section{Introduction} \label{sec:intro}

For models with large numbers of parameters and large volumes of data the need to sample posterior distributions for exact Bayesian inference has, for many, rendered the perception of traditional Bayesian approaches as too slow for practical use \citep{knoblauch2022optimization}. However, Bayes provides a model for learning that is foundationally coherent \citep{goldstein2013observables, williamson2015posterior} and represents the gold standard for providing uncertainty quantification in machine learning algorithms \citep{murphy2012probabilistic, blundell2015weight,gal2016dropout,wenzel2020good,matsubara2021ridgelet,tran2022all,wang2023data}. For these reasons there has been a great deal of recent interest in \textit{Bayes-as-optimisation} approaches, seeking to provide Bayesian inference via, in principle, faster optimisation methods rather than slower sampling approaches. 

Motivated by computational expediency, variational approaches to Bayesian inference have been studied for many decades \citep{blei2017variational, jordan1999introduction}. These approaches typically make simplifying approximations to the form of the posterior, often amounting to posterior independence of the model parameters. These approximations enable the posterior of each parameter to be expressed in terms of the variational distribution and then optimised for efficient approximate inference. More recently, the exact posterior was shown to be the solution to an optimisation problem involving minimising the expected loss of a given loss function \citep{bissiri2016general}. This opened an area of research termed Generalised Bayesian Inference (GBI) wherein alternative loss functions were explored to allow more robust inferences in certain situations \citep{bissiri2016general,jewson2018principles,knoblauch2022optimization,jewson2023stability}. We expand on this theory more thoroughly in Section \ref{sec:GBIandOpt}.

Within the GBI literature, Bayes linear methods \citep{goldstein2007bayes} are often cited as an example of a generalised approach to Bayesian inference, though their interpretation nor how they fit within GBI is never made clear. Certainly, Bayes linear methods are a class of \textit{Bayes-as-optimisation} approaches, but their relation to modern theory has not been clearly established, with more recent accounts deferring to a short explanation offering how to compute a Bayes linear update, before moving on to the focus of their application and any specific innovations therein \citep{vernon2022bayesian, xu2021local, jackson2023efficient}. In fact, in order to establish Bayes linear approaches as a form of GBI, we must appeal to the underlying geometry of the Bayes linear paradigm. In this paper, we tackle this comparison with the aim of attracting more attention to Bayes linear ideas within the GBI community.

In addressing the connection between GBI and Bayes linear methods, we take inspiration from the GBI literature and propose a generalisation to Bayes linear approaches that allows for fast inference on restricted classes of quantities. We term this methodology \textit{generalised Bayes linear inference} and demonstrate its efficacy on a number of examples, including monotonic emulation and inference for spatial counts. The paper is structured as follows. Section \ref{sec:BL} presents a traditional account of Bayes linear statistics. Section \ref{sec:GBIandOpt} describes GBI and proposes a wider definition of a generalised Bayesian inference to encapsulate Bayes linear methods and other recent geometric approaches to Bayesian inference. Section \ref{sec::geometric} frames Bayes linear inference as a principled optimisation problem within a specific measure-free geometry and as a specific example of a method fitting within our wider definition of GBI. Section \ref{sec:gbli} generalises Bayes linear inference to problems with restricted solution spaces and Section \ref{sec:applications} presents a collection of applications of our method with comparison to standard approaches. We conclude in Section \ref{sec:discuss}. This work is accompanied by an \texttt{R} package available at \url{github.com/astfalckl/bayeslinear}.

\section{Bayes linear statistics}\label{sec:BL}

A traditional exposition for Bayes linear methods involves learning about a random quantity $\bX \in \mathbb{R}^n$, using a collection of data $\bD = (D_1, \dots, D_m)$, where $D_i$ is the $i$th observed datum. In his seminal work, \cite{de2017theory} argued that expectation should be the primitive quantity used to express belief and derived his theory of probability through expectation. De Finetti laid the ground work for partial prior specification, implying that one could specify the means, variances and covariances over $\bX$ and $\bD$ without reference to an underlying probability measure. Bayes linear methods, first published in \cite{goldstein1981revising} and most notably in \cite{goldstein2007bayes}, endow this theory with the required machinery for inference and belief revision. Equipped with $\E[\bX]$, $\var[\bX]$, $\E[\bD]$, $\var[\bD]$ and $\cov[\bX, \bD]$, beliefs about $\bX$ can be updated by $\bD$ using the Bayes linear updating equations:
\begin{linenomath}
\begin{align}
     \mathbb{E}_\bD[\bX] &= \mathbb{E}[\bX] + \cov[\bX,\bD] \var[\bD]^\dagger(\bD - \mathbb{E}[\bD]),\label{eqn:adj_exp} \\
      \var_\bD[\bX] &= \var[\bX] - \cov[\bX,\bD] \var[\bD]^\dagger \cov[\bD,\bX],\label{eqn:adj_var}
\end{align}
\end{linenomath}
where $\var[\bD]^\dagger$ is any pseudo-inverse of $\var[\bD]$, most commonly the Moore-Penrose inverse.
The updated quantities, $\mathbb{E}_\bD[\bX]$ and $\var_\bD[\bX]$ are known as the \textit{adjusted expectation} and \textit{adjusted variance} respectively. The adjusted expectation, $\E_{\bD}[\bX]$, is the solution to the optimisation problem that seeks to find the affine combinations of $\bD$, $\bh + \bH\bD$, that minimise expected squared-error loss, $\E[(\bX - (\bh + \bH\bD))^2]$. When we observe some $\bd$, we can evaluate the random quantity in \eqref{eqn:adj_exp}, known as the observed adjusted expectation, $\mathbb{E}_\bd[\bX]$. As we show in Theorem~\ref{the:bl}, Equations~\eqref{eqn:adj_exp} and \eqref{eqn:adj_var} are identical to the posterior moments of a probabilistic Bayesian analysis for a broad class of problems.

\begin{theorem} \label{the:bl}
    In a probabilistic Bayesian analysis, if the posterior expectation, $\mathbb{E}[\bX \mid \bD]$, is affine in $\bD$ such that $\mathbb{E}[\bX \mid \bD] = \bA\bD + \bB$, then the posterior expectation, $\mathbb{E}[\bX \mid \bD]$, and variance, $\mathrm{var}[\bX \mid \bD]$ have the functional forms of \eqref{eqn:adj_exp} and \eqref{eqn:adj_var}, respectively.
\end{theorem}
\begin{proof}
    Proof is provided in Appendix~B.
\end{proof}

\begin{remark} \label{rem:bl}
    The Bayes linear update is often considered to only be analogous to fully Bayesian specifications for Gaussian distributed data due to the well-known similarities between \eqref{eqn:adj_exp} and \eqref{eqn:adj_var} and the posterior distribution calculated assuming a Gaussian likelihood and prior. As noted in \cite{ericson1969note}, the class of likelihood and prior distribution pairings for which \eqref{eqn:adj_exp} describes the posterior mean is much larger than just the Gaussian family. Indeed, the condition in Theorem~\ref{the:bl}, that the posterior is linear in $\bD$, holds for the exponential family of distributions with conjugate prior, also further discussed in Appendix~B. Noting this should broaden the acceptance of using Bayes linear inference for more varied data, including that which exhibits logistic, count, and skewed behaviour.
\end{remark}
 
In order to demonstrate the connections between modern Bayesian approaches and Bayes linear methods, we now describe the core ideas of GBI and then give a geometric exposition for the Bayes linear paradigm that draws the relevant connections to GBI approaches.

\section{Generalised Bayesian Inference}\label{sec:GBIandOpt}

\cite{bissiri2016general} define GBI as solving 
\begin{equation} \label{eqn:GBI}
    \textit{q}^{*}(\boldsymbol{\theta}) = \underset{q\in\Pi}{\mathrm{argmin}} \Bigg\{\mathbb{E}_{\textit{q}(\boldsymbol{\theta})} \Bigg[\sum_{i=1}^{n}\ell(\boldsymbol{\theta},x_i)\Bigg] + \mathrm{KLD}(\textit{q}\parallel\pi)\Bigg\}, 
\end{equation} 
where $\textit{q}^{*}(\boldsymbol{\theta})$ is a probability measure over an unknown parameter $\boldsymbol{\theta}$, the $x_i$ are the observed data, KLD denotes the Kullback-Liebler divergence, $\pi$ is a prior measure on $\boldsymbol{\theta}$ and $\Pi$ is the set of all probability distributions with density $q(\boldsymbol{\theta})$. Note that here, and in the remainder of this section, we alter the nomenclature of the unknown quantity on which we infer from $\bX$ to $\boldsymbol{\theta}$ and the observed data from $D_i$ to $x_i$. We do so to align with the GBI literature that we review. The first term in \eqref{eqn:GBI} is a finite-sample approximation to the expected loss, 
\begin{equation*}
  \E[\ell(\boldsymbol{\theta},x)] =  \int\int \ell(\boldsymbol{\theta},x) \; \mathrm{d}F_0(x)q(\mathrm{d}\boldsymbol{\theta}), 
\end{equation*}
where $F_0(x)$ is the data generating process. Parameter inference takes place in the same geometric space associated with traditional Bayesian analysis, $\mathcal{L}_{2}(\Theta,\mu)$, with weighted integral inner product with respect to measure $\mu$, $\langle f(\boldsymbol{\theta}), g(\boldsymbol{\theta}) \rangle = \int_\Theta f(\boldsymbol{\theta})g(\boldsymbol{\theta})\; \mu(\mathrm{d} \boldsymbol{\theta})$.

In choosing $\ell$ to be the negative log-likelihood, \eqref{eqn:GBI} recovers the traditional Bayesian posterior and defines the quantity of interest $\boldsymbol{\theta}$ to be the parameter minimising the KLD between $F_0(x)$ and the model parametrised by $\boldsymbol{\theta}$, $f(x;\boldsymbol{\theta})$. \cite{bissiri2016general}, \cite{jewson2018principles} and \cite{knoblauch2022optimization} produce \textit{generalised} updates through the derivation and use of alternative scoring rules to motivate loss functions which target the minimisation of more robust divergences (relative to the KLD) between $F_0(x)$ and $f(x;\boldsymbol{\theta})$: 
\begin{linenomath}
\begin{equation*}
    \text{D}(F_0||f) = \E_{F_0}[S(x,F_0)] - \E_{F_0}[S(x,F)], 
\end{equation*}
\end{linenomath} 
where $S$ is user-defined scoring rule. For example, see \cite{jewson2023stability} for a recent study of the properties of generalised Bayesian updates using the $\beta$-divergence loss function. \cite{knoblauch2022optimization} generalise further by taking the second term within the optimisation in \eqref{eqn:GBI} to be general divergence $D(\cdot||\cdot)$ and restricting the solution space to families of tractable distributions $\mathcal{Q}$. Updates whereby the solution space is restricted are generally referred to as Generalised Variational Inference (GVI). For examples of recent discussion and extensions of these ideas see \cite{pacchiardi2021generalized,wild2022generalized,husain2022adversarial,matsubara2022robust,matsubara2023generalized,jewson2023stability}.

GBI, as defined above, recovers a host of different variations of Bayesian methodology; however, it does not fully encapsulate the available methodology by which one can produce optimal posterior beliefs that do not necessarily coincide with the traditional measure-theoretic Bayesian posterior. There are approaches to Bayesian inference via optimisation such as Bayes linear method or the method of \cite{de2019geometry} (given below) that don't fit so readily within the framework. We believe a further geometrically motivated generalisation covers a much wider class of methods and suggests avenues for further innovation. It also highlights the similarities and differences between GBI, GVI and Bayes linear methods, and we advocate for these methodologies to be seen within our definition.
\begin{definition}
    The three properties of a generalised Bayesian inference are
    \begin{itemize}[leftmargin=8em]
    \item[\textbf{Property 1}] The unknown quantity $\boldsymbol{\theta}$ and the data $x$, are defined on a separable completely metrisable topological space, $\mathcal{G}$, i.e. $\boldsymbol{\theta}$ and $x$ form a Polish space \citep[see][for a rigorous treatment of Polish spaces]{kechris2012classical}.
    \item[\textbf{Property 2}] The space $\mathcal{G}$ is equipped with a statistical divergence $d$, such that for $a, b \in \mathcal{G}$, $d(a, b) \geq 0$ and $d(a, b) = 0$ if and only if $a = b$.
    \item[\textbf{Property 3}] 
    An optimisation with respect to $\boldsymbol{\theta}$, within solution space $\boldsymbol{\theta} \in \mathcal{C} \subset \mathcal{G}$, for the belief representation that is closest to the data generating process with respect to the statistical divergence $d$.
\end{itemize}
\end{definition}
In a standard probabilistic Bayesian analysis, $\mathcal{G}$ is further equipped with a Borel $\sigma$-algebra $\mathcal{B}_\sigma$ and $\sigma$-finite dominating measure $\mu$ so that $(\mathcal{G}, \mathcal{B}_\sigma, \mu)$ forms a $\sigma$-finite measure space. As we will explore in Section~\ref{sec::geometric}, this is not a necessary assumption under which to conduct inference. Further, additional restrictions on the statistical divergence $d$ may impose additional structure in $\mathcal{G}$. If for all $a \in \mathcal{G}$, the Taylor expansion of $d(a, a + \mathrm{d}a)$ has a positive definite quadratic form, then $d$ defines a Riemann metric $g$ on $\mathcal{G}$ so that $(\mathcal{G}, g)$ forms a Riemannian manifold. For example, this is the case when $d$ is the Kullback-Leibler divergence, but not the total variation distance. If $d(a, b) = d(b, a)$ and satisfies the triangle inequality $d(a, c) <= d(a, b) + d(b, c)$ then $(\mathcal{G}, d)$ is a complete metric space with metric $d$. If the metric is calculated from a norm, so that $d(a + h, b + h) = d(a, b)$ and $d(ha, hb) = |h|d(a, b)$, then $(\mathcal{G}, d)$ is a Banach space, as is the case in Bayes linear statistics.

We also note our use of the term \textit{data generating process} in Property~3 is not meant to signify the requirement for the existence of a random process that generates the data. Even within GBI, there are interpretations of the data generating process, $F_0(x)$, as a truly held subjective probability distribution that you are unable to specify due to time and resource constraints \citep{jewson2018principles, bird2023meaning}. We use the term as it is familiar to most researchers, yet its meaning will depend on the choices made for Properties 1 and 2. For Bayes linear methods, Properties 1 and 2 lead to measure-free inference, so the data generating process is not required to be a probability distribution.

Definition~1 clearly captures the framework associated with GBI and GVI, where the geometry, $\mathcal{G} = \mathcal{L}_{2}(\Theta,\mu)$, Property 2 tends to be met via a chosen divergence and the optimisation in Property 3 is, in GBI, either with respect to the family of posterior distributions $\Pi$ imposed by the likelihood and prior,   or in GVI some pre-specified family of distributions $\mathcal{Q}$. An example of a recent approach that does not fit within previous definitions of GBI is \cite{de2019geometry}. Here $\mathcal{G}$ is taken to be $\mathcal{L}_{2}(\Theta, \mu)$ with uniform measure $\mu(\boldsymbol{\theta})$. The inner product is used to define the compatibility measure 
\begin{equation*}
    \kappa_{g,h} = \frac{\langle g,h \rangle}{\|g\|\|h\|}\;\;\;g,h \in \mathcal{L}_{2}(\Theta),
\end{equation*}
making their notion of closeness the level of compatibility captured by $\kappa_{g,h}$, imposing a distance $d(g, h) = 1 - \kappa_{g,h}$ that satisfies Property~2. Compatibility is maximised between functions $g$ and $h$, where these elements can be taken to be different combinations of potential prior, posterior and likelihood functions. Updates of this kind, potentially involving optimisation of metrics such as that of compatibility in order to produce alternative posteriors fall comfortably within Definition~1.

Generally, $\mathcal{G}$ will be an inner product space and most probabilistic methods will represent closeness via some divergence; however, even in these settings this represents a particular choice. As an example, \cite{informationgeometry} argue that one can choose to represent the space of Gaussian distributions as a two-dimensional manifold with coordinate system ($\mu,\sigma$); this is topologically equivalent to the upper-half of a two-dimensional Euclidean space. One could therefore define closeness between two points on the manifold of Normal distributions as a distance in Euclidean space.

\section{The geometry of Bayes linear statistics}\label{sec::geometric}

\subsection{The belief structure}

We now give a formal geometric development of Bayes linear statistics in order to establish the theory as a generalised Bayesian inference, as per Definition~1. We revert back to the notation familiar in the Bayes linear literature and define $\bX$ as our random quantity and $\bD$ as the collection of the observed data. First, we establish the underlying geometry (Property 1). In general, random quantities are defined to live in a vector space, $X$, equipped with an inner product to form the Hilbert space, $\mathcal{X}$, with the standard assumption of Cauchy completeness in the corresponding Banach space (this is tantamount to assuming that the random quantities have finite variance). Rather than the $\mathcal{L}_2$ inner product with measure $\mu$, the Bayes linear paradigm uses an alternative inner product that enables expectation, rather than probability, to be the primitive quantity. Consider the random quantity, $\bX \in \mathbb{R}^n$, defined on the Hilbert space $\mathcal{X}$ endowed with the \textit{product inner product} $\langle \bX, \bY \rangle = \mathbb{E}[\bX^\T\bY]$. Assume, herein, all defined Hilbert spaces to be equipped with the product inner product. As before, denote by $D_i$ the $i$th datum and the collection $\bD = (D_1, \dots, D_m)$ as all of the data; further define $\mathcal{D}$ as the Hilbert space in which $\bD$ lives. In standard Bayesian inference a joint probability measure is required to be specified for $\bX$ and $\bD$; in Bayes linear statistics we are required to define the analogous joint inner product space $\mathcal{B} = \mathcal{X} \cup \mathcal{D} \cup \textcal{1}$, also known as the \textit{belief structure}. Here, $\textcal{1}$ denotes the linear subspace spanned by the unit constant $\one$ and is required to identify random variables that differ by a constant. Construction of $\mathcal{B}$ requires the specification of the totality of norm expressions between the elements of $\bX$, $\bD$ and $\textcal{1}$ and so requires specification of $\mathbb{E}[\bX]$, $\mathbb{E}[\bD]$, $\var[\bX]$, $\var[\bD]$, and $\cov[\bX, \bD]$. 

As soon as we have fully specified our belief structure, $\mathcal{B}$, we have adequate functionality to describe notions of distance and similarity (Property 2). For example, if we consider the individual elements of $\bX$ and $\bD$ to be zero-mean and form a basis in $\mathcal{B}$, we may describe the similarity of any linear combination of $\bX$ and $\bD$, $a^\T \bX$ and $b^\T \bD$ as $\langle a^\T \bX, b^\T \bD \rangle = a^\T \cov[\bX, \bD] b$. We may also extend this to non-random vectors in $\mathcal{B}$; for example, data $\bd \in \mathbb{R}^m$ observed in $\mathcal{D}$. If we represent $\bd$ via a basis representation in $\mathcal{D}$, $\bd = b^\T \bD$, we can calculate the norm, as above, as $\norm{\bd}^2 = b^\T \var[\bD] b$. Observing $\bd$ gives rise to a linear functional that can be used in Riesz's representation theorem to identify the coordinate vector $b = \var[\bD]^\dagger \bd$ \citep[see][Chapter 4]{goldstein2007bayes}. Consequently, the squared distance between any two observed elements, $\bd_1, \bd_2$ in $\mathcal{D}$ can be expressed as the Mahlanobis distance $\|\bd_1 - \bd_2\|^2 = (\bd_1 - \bd_2)^\T\var[\bD]^{\dagger}(\bd_1-\bd_2)$.

\subsection{Adjusting belief structures}\label{sec::adj_belief_struc}

Bayes linear inference projects the random quantity of interest into subspaces of $\mathcal{B}$ defined by the observed quantities. Projection is a specific optimisation with the chosen subspaces representing the restricted solution space (Property 3). Consider an observable $\bF \in \mathbb{R}^q$. The adjusted expectation, $\mathbb{E}_\bF[\bX]$, of $\bX$ by $\bF$, is the orthogonal projection of $\bX$ onto $\{\bH \bF \mid \bH \in \mathbb{R}^{n \times q}\}$ a linear subspace of $\mathcal{B}$, that minimises $\|\bX - \mathbb{E}_\bF[\bX]\|$. The linear subspace defined by the quantity $\bX - \mathbb{E}_\bF[\bX]$ describes the aspects of $\bX$ unresolved by $\bF$. After the adjustment of $\bX$ by $\bF$, the inner product remains unchanged (the product inner product) but with revised random quantities: $\langle \bX, \bY \rangle_\bF = \langle \bX - \mathbb{E}_\bF[\bX], \bX - \mathbb{E}_\bF[\bY] \rangle = \mathbb{E}\left[\left(\bX - \mathbb{E}_\bF[\bX]\right)^\T\left(\bY - \mathbb{E}_\bF[\bY]\right)\right]$. The simplest adjustment to consider is the adjustment of a random quantity $\bX$ by the unit constant, $\bm{1}$. Here, $\mathbb{E}_{\bm{1}}[\bX] = \mathbb{E}[\bX]$. After this adjustment, the unresolved aspects of $\bX$ are $\bX - \mathbb{E}[\bX]$ and the inner product $\langle \bX, \bY \rangle_{\bm{1}} = \mathbb{E}\left[\left(\bX - \mathbb{E}[\bX]\right)^\T\left(\bY - \mathbb{E}[\bY]\right)\right]$ is equivalent to the covariance inner product. Consequently, distance in the inner product space $\langle \cdot, \cdot \rangle_{\bm{1}}$ now has a direct interpretation with uncertainty. If we adjust the belief space $\mathcal{B}$ by the unit constant, then the outer product that defines the space is wholly represented by the covariance matrices $\var[\bX]$, $\var[\bD]$, and $\cov[\bX, \bD]$. 

The adjusted expectation used in practice projects $\bX$ onto the joint linear subspace represented by the constant $\bm{1}$ and the observed data $\bD$ so that 
\begin{equation} \label{eqn:adj_exp_def}
    \mathbb{E}_{(\bm{1}, \bD)}[\bX] = \begin{pmatrix}\bh, \bH \end{pmatrix} \begin{pmatrix}\bm{1} \\ \bD \end{pmatrix} = \bh + \bH \bD \text{ for } \bh \in \mathbb{R}^{n \times 1} \text{ and } \bH \in \mathbb{R}^{n \times m}.
\end{equation}
Note that $\mathbb{E}_{(\bm{1}, \bD)}[\bX] = \mathbb{E}_{(\bD - \mathbb{E}[\bD])}[\bX - \mathbb{E}[\bX]]$. As, to identify a unique solution, we must always project onto $\textcal{1}$, it is traditionally suppressed in the notation so that both are written $\mathbb{E}_{\bD}[\bX]$ and the inner product is taken as the covariance inner product. Minimising the distance $\norm{\bX - \mathbb{E}_\bD[\bX]}$ for $\mathbb{E}_\bD[\bX]$ as in (\ref{eqn:adj_exp_def}) yields the adjusted expectation as given in \eqref{eqn:adj_exp}. Our remaining uncertainty in $\bX$, after adjusting by $\bD$, is represented by the linear subspace $\bX - \mathbb{E}_\bD[\bX]$. The adjusted variance $\var_\bD[\bX]$ represents the totality of the norm expressions in $\bX - \mathbb{E}_\bD[\bX]$, and is defined by the outer product $\mathbb{E}\left[\left(\bX - \mathbb{E}_\bD[\bX]\right)\left(\bX - \mathbb{E}_\bD[\bX]\right)^\T\right]$, which can be readily computed via \eqref{eqn:adj_var}. Note that $\var_\bD[\bX]$ is unaffected by the actual observation of $\bD$ and so is in effect a prior inference. Derivations of \eqref{eqn:adj_exp} and \eqref{eqn:adj_var} are given in Appendix~A; these are alternative derivations to the often cited derivations in \cite{goldstein2007bayes} and provide an alternative insight to inference as orthogonal projections.

Note that, for the Bayes linear paradigm, the `data generating process' is just $\bX$, the value of the random quantity that coincides with data, $\bD$. As we only specify $\cov[\bX, \bD]$, there is no model for the data conditional on the random quantity and it would be imprecise to say that $\bX$ \textit{generates} $\bD$, though there may be examples where causality is assumed within a Bayes linear analysis. We use the term `data generating process' only to clarify the parallels with more mainstream statistical approaches.

\subsection{Bayes linear as GBI}

Considering Bayes linear methods as a form of GBI offers avenues to enrich both methodologies. Perhaps what makes Bayes linear approaches seem so unusual is the alternative choice of the underlying geometry (Property 1, i.e. the product inner product). In Bayes linear statistics, this choice of inner product naturally facilitates a notion of distance (Property 2), whereas divergences and their specification are the focus of much of the current GBI literature. \cite{klebanov2021linear} propose to extend Bayes linear methods via a fully probabilistic study of the `linear conditional expectation' in infinite-dimensional Hilbert space. However, this changes the geometry (Property 1), and defines a different notion of closeness and is thus perhaps best viewed as another form of GBI, rather than an extension of Bayes linear methods. 

Restricting the solution space to some $\mathcal{C}$ (Property 3) has received significantly less attention in Bayes linear inference and will be of particular interest for the remainder of this article. Though projection onto the linear subspace spanned by $\{\one, \bD\}$ represents a particular restriction for Bayes linear methods, those subspaces are not constrained to particular subsets in the original Hilbert space in which the random quantity is defined. For example, if our random quantity is a strictly non-negative quantity, there is nothing to ensure that the adjusted expectation respects this. Such restrictions can be handled trivially within the geometries more frequently used for GBI, either by setting prior measures to zero via Property 2, or restricting the solution space to measures that are zero wherever restrictions are required via Property 3. In the following section we develop a generalised Bayes linear inference that enables inference to take place within subsets in Hilbert spaces. 

\section{Generalised Bayes linear inference}\label{sec:gbli}

\subsection{The generalised adjusted expectation}

In many statistical inference problems we want to restrict the solution space to ensure our adjusted expectation $\mathbb{E}_\bD[\bX] \in \mathcal{C} \subseteq \mathbb{R}^n$. Where the adjusted expectation is subject to such a constraint we will use the notation $\mathbb{E}_\bD^\mathcal{C}[\bX]$. For example, we may want to restrict our inference to the positive cone (random quantities are all positive), the isotonic cone (random quantities are monotonically ordered), or the positive semi-definite cone (random quantities form a valid variance matrix). Note that $\mathbb{E}_\bD^\mathcal{C}[\bX]$ is not necessarily affine in $\bD$; we arrive at this fact via example. If we assume $\mathbb{E}_\bD^\mathcal{C}[\bX] = \bh_0 + \bH_0 \bD$, then, having observed $\bd$, computed $\mathbb{E}_\bd[\bX]$ and found that $\mathbb{E}_\bd[\bX] \notin \mathcal{C}$ we could reason that we need to solve for some other $\bh$ and $\bH$ than what is given by Bayes linear theory. This decision to re-calculate $\bh$ and $\bH$ is dependent on what value we observed from $\bD$. Consequently, $\bh$ and $\bH$ are functions of $\bD$ and so the assumption of linearity does not hold. 

To allow for a restricted solution space, $\mathbb{E}_\bD^\mathcal{C}[\bX] \in \mathcal{C}$, we augment the traditional Bayes linear theory. Define $\mathcal{X}_{\bD}$ as the Hilbert space, still endowed with the product inner product but adjusted by $\bD$, so that the inner product is 
\begin{equation*}
    \langle \bX, \bY \rangle_\bD = \langle \bX - \mathbb{E}_\bD[\bX], \bY - \mathbb{E}_\bD[\bY] \rangle.
\end{equation*} 
This is still the same original space and geometry as in Property 1, but with the adjusted belief space given $\bD$ described in Section \ref{sec::adj_belief_struc}. The observed adjusted expectation, $\mathbb{E}_\bd[\bX] \in \mathcal{X}_{\bD}$, does not necessarily fall in $\mathcal{C} \subseteq \mathcal{X}_\bD$. We propose a method that respects the original belief specification as being truly held, by using it, and the adjustment, to define the inner product space. To amend the inference accordingly, we then restrict the solution space to $\mathcal{C}$. We define the solution $\mathbb{E}_\bd^\mathcal{C}[\bX]$ as the orthogonal projection of $\mathbb{E}_\bd[\bX]$ into $\mathcal{C}$, found by computing 
    \begin{equation} \label{eqn:objective}
      \mathbb{E}_\bd^\mathcal{C}[\bX] = \underset{\bq \in \mathcal{C}}{\arg \min} \| \mathbb{E}_\bd[\bX] - \bq\|.
    \end{equation}
As discussed in Section \ref{sec::geometric}, the norm in \eqref{eqn:objective} is computed as
    \begin{equation} \label{eqn:norm}
        \| \mathbb{E}_\bd[\bX] - \bq\|^2 = (\mathbb{E}_\bd[\bX] - \bq)^\T \var_\bD[\bX]^{-1} (\mathbb{E}_\bd[\bX] - \bq),
    \end{equation}
noting that the inner product remains defined by the adjusted belief specification. As noted above, imposing a solution space and so altering Property 3, within our specific geometry, is the most natural way in which we can exert a comparative degree of control over our solution space relative to other GBI methodologies.

\begin{remark}
    We define this theory for general subsets $\mathcal{C}$. If we assume $\mathcal{C}$ to be convex, the Hilbert projection theorem states that there exists a unique vector $\bq$ for which \eqref{eqn:objective} is minimised, over all vectors $\bq \in \mathcal{X}_\bD$. Further, as \eqref{eqn:norm} is a convex function, as long as $\mathcal{C}$ is not the empty set, a solution is guaranteed. We note that conic projections are a sub-class of convex optimisation problems, and so this guarantee of existence and uniqueness also holds for projections into conic solution spaces.
\end{remark}

If $\mathbb{E}_\bd[\bX] \notin \mathcal{C}$, retaining the adjusted expectation as our own expectation having observed the data violates coherence. Whether the source of the incoherence is our original belief specification, the restriction of the solution space to $\mathrm{span}\{\one, \bD\}$, or the combination of these two, revisiting either or both to address the problem is highly challenging in practice. Even the most experienced Bayes linear practitioners would struggle to make foundationally meaningful adjustments to $\cov\left[\bX, \bD\right]$, $\var\left[\bX\right]$ and $\var\left[\bD\right]$, even for univariate $\bX$, that would ensure a coherent adjusted expectation. The restriction to $\mathcal{C}$ is likely the most firmly held of the prior statements an expert or practitioner might make. However, the partial prior specification required for Bayes linear methods can only loosely reflect it through ensuring that $\mathbb{E}[\bX]\in \mathcal{C}$. It can still be the case that the second order prior specification is also believed and that there are possible data that lead to $\mathbb{E}_\bd[\bX] \notin \mathcal{C}$ (see illustrative example in Section \ref{sec:example}). In this case, the source of the incoherence is the restriction of the solution space to $\mathrm{span}\{\one, \bD\}$,  and projecting into explicitly non-linear (in \bD) solution spaces would require further or different prior specifications as well as a different theory.
    
An important motivation for retaining the adjusted beliefs for determining the inner product, is to ensure that the sequential nature of Bayesian updating is retained. Namely, the adjustment by $\bD_1$ and then by $\bD_2$ should be equivalent to the adjustment by $\bD_1\cup \bD_2$. As the adjustment happens in the usual fashion, amending the inner product by $\mathbb{E}_\bD[\bX]$, this property is guaranteed by the usual Bayes linear theory \citep[see][Chapter 5]{goldstein2007bayes}. Given observations, the inference, or prediction, we report would then be restricted to $\mathcal{C}$ via (\ref{eqn:objective}). The same happens in GBI approaches, particularly those with restricted solution spaces such as variational inference. For example, the variational approximation is not found (say using $\bD_1$) and then itself updated by further data (say $\bD_2$). Instead, the posterior given $\bD_1$ is approximated and, given more data, it is then approximated again using $\bD_1\cup\bD_2$.

\subsection{The generalised adjusted variance} \label{sec:gen_adj_var}

The adjusted variance is defined from the outer product $\mathbb{E}[(\bX - \mathbb{E}_\bD[\bX]) (\bX - \mathbb{E}_\bD[\bX])^\T]$ assuming the affine representation $\mathbb{E}_\bD[\bX] = \bh_0 + \bH_0 \bD$. As when $\mathbb{E}_\bd^\mathcal{C}[\bX] \neq \mathbb{E}_\bd[\bX]$, our adjusted beliefs are not affine in $\bD$ and this definition of adjusted variance does not extend to our generalised setting. Rather, we propose an alternative definition that suits some additional coherency arguments. To aid our developments, we represent the vector $(\mathbb{E}_\bd^\mathcal{C}[\bX] - \mathbb{E}_\bd[\bX]) \in \mathcal{X}_\bD$ in an orthonormal space $\mathcal{Z}$, so that $\var[\bZ] = \bI$, by defining 
\begin{equation}
  \bz = (z_1, \dots, z_n)^\T = \bL^\dagger (\mathbb{E}_\bd^\mathcal{C}[\bX] - \mathbb{E}_\bd[\bX]),
\end{equation}
where $\var_\bD[\bX] = \bL \bL^\T$ is any square-root decomposition of $\var_\bD[\bX]$. We name the quantity, $\bz\in \mathcal{Z}$, the \textit{constraint discrepancy}. When $\bL = \bQ \sqrt{\blambda}$, where $\blambda$ is a diagonal matrix of the ordered eigenvalues of $\var_\bD[\bX]$ and $\bQ$ stores the corresponding eigenvectors, we obtain a representation whereby the elements $(z_1, \dots, z_n)$ describe the standardised distance between $\mathbb{E}_\bd^\mathcal{C}[\bX]$ and $\mathbb{E}_\bd[\bX]$ in terms of the principal axes of $\var_\bD[\bX]$. We can use any square-root decomposition to project onto an orthonormal basis; however, using $\bL = \bQ \sqrt{\blambda}$ provides an appealing interpretation. 

We define the generalised adjusted variance, $\var_\bd^\mathcal{C}[\bX]$, as a rotation of $\bS$, where $\bS$ is a diagonal matrix with each entry, $\bS_{ii}$, defined by a function of the $i$th constraint discrepancy, $f(z_i)$, and $\var_\bd^\mathcal{C}[\bX] = \bL \bS \bL^\intercal$. We define three statements/constraints of coherency that restrict and motivate the choice of $f(z_i)$:
\begin{itemize}[leftmargin=10em]
    \item[\textbf{Constraint 1:}] The limit $\underset{|z_i| \rightarrow 0}{\text{lim}} \left( \bS_{ii} \right) = 1$.
    \item[\textbf{Constraint 2:}] The limit $\underset{|z_i| \rightarrow \infty}{\text{lim}} \left( \bS_{ii} \right) = 0$.
    \item[\textbf{Constraint 3:}] $f(\cdot)$ is a non-increasing function of $z_i.$
\end{itemize}
Constraint 1 ensures that our method generalises Bayes linear methods so that if $\mathbb{E}_\bd^\mathcal{C}[\bX] = \mathbb{E}_\bd[\bX]$, then $\var_\bd^\mathcal{C}[\bX] = \var_\bD[\bX]$, and so when the adjusted expectation coheres with $\mathbb{E}_{\bd}[\bX] \in \mathcal{C}$, we retain the standard belief revisions. Constraints 2 and 3 codify the notion that the further away the standard adjusted expectation is from $\mathcal{C}$, the more sure we ought to be that the random quantity, $\bX$, is the nearest point in $\mathcal{C}$ to $\mathbb{E}_{\bd}[\bX]$. To see this consider an adjusted expectation just outside of $\mathcal{C}$, but where $\var_{\bD}[\bX]$ implies that many points within $\mathcal{C}$ are credible values of $\bX$. In this case, the generalised adjusted variance should only be marginally smaller than the adjusted variance. As $\mathbb{E}_{\bd}[\bX]$ moves further away from $\mathcal{C}$ the subset of credible values according to the original $\var_{\bD}[\bX]$ reduces, and so $\var_{\bd}^\mathcal{C}[\bX]$ should reduce accordingly.

Specification of the function $\mathrm{f}(\cdot)$ that describes $\bS_{ii} = \mathrm{f}(z_i)$, is an additional belief specification. For instance, the rate of decay in $\mathrm{f}(\cdot)$ will be slower for random variables with heavier tails, that is, larger kurtosis. In the absence of a strong belief that can be translated to a choice of $\mathrm{f}(\cdot)$, we propose 
\begin{equation} \label{eqn:cantelli}
  f(z_i) =\bS_{ii} = \frac{1}{1 + z_i^2},
\end{equation}
motivated by Cantelli's inequality, the one-sided improvement to Chebyshev's inequality, that states an upper bound on the upper tail probability for random quantity $Z_i$ being larger than its observed value $z_i$. Whilst Bayes linear methods don't admit probabilistic belief specification, they regularly make use of such probabilistic inequalities to guide diagnostics in model selection \citep[see][]{goldstein2007bayes, craig1996bayes, salter2022efficient}.

\subsection{An illustrative 2D example} \label{sec:example}

We illustrate our developments in this section with a simple example that also serves to demonstrate how easy it is for standard Bayes Linear methods with reasonable prior specifications to provide incoherent results. We define $\bX$ and $\bD$ as unit-mean bivariate quantities with a belief structure defined by
\begin{equation*}
    \mathrm{var}[\bX] = \begin{bmatrix}
        0.54 & 0.09 \\ 0.09 & 0.54
    \end{bmatrix}, \quad
    \mathrm{var}[\bD] = \begin{bmatrix}
        1 & -0.2 \\ -0.2 & 1
    \end{bmatrix}, \quad
    \mathrm{cov}[\bX, \bD] = \begin{bmatrix}
        0.4 & -0.1 \\ -0.1 & -0.3
    \end{bmatrix}.
\end{equation*}
We further define the solution space $\mathcal{C} = \{\bX \in \mathbb{R}^2 \mid \bX \succeq 0\}$ as the non-negative cone defined from the non-negative orthant in $\mathcal{X}$. Given an observation $\bd = (3, 6.5)$, calculating the standard Bayes linear updates, as given in \eqref{eqn:adj_exp} and \eqref{eqn:adj_var}, yields observed adjusted expectation and adjusted variance
\begin{equation*}
    \mathbb{E}_\bd[\bX] = \begin{bmatrix}
        1.68  \\ -1.17
    \end{bmatrix}, \quad
    \mathrm{var}_\bD[\bX] = \begin{bmatrix}
        0.38 & 0.123 \\ 0.123 & 0.423
    \end{bmatrix}.
\end{equation*}
These results are shown in the left of Figure~\ref{fig:2d_simulation}. The observed adjusted expectation falls outside of $\mathcal{C}$ and so the generalised adjusted expectation, $\mathbb{E}_\bd^\mathcal{C}[\bX]$, and variance, $\mathrm{var}_\bd^\mathcal{C}[\bX]$, will differ from $\mathbb{E}_\bd[\bX]$ and $\mathrm{var}_\bD[\bX]$. We calculate $\mathbb{E}_\bd^\mathcal{C}[\bX]$ from \eqref{eqn:objective} and $\mathrm{var}_\bd^\mathcal{C}[\bX]$ as described in Section~\ref{sec:gen_adj_var} with the additional belief specification of variance decay following \eqref{eqn:cantelli}. This yields the generalised belief updates
\begin{equation*}
    \mathbb{E}_\bd^\mathcal{C}[\bX] = \begin{bmatrix}
        2.02 \\ 0
    \end{bmatrix}, \quad
    \mathrm{var}_\bd^\mathcal{C}[\bX] = \begin{bmatrix}
        0.106 & 0.0344 \\ 0.0344 & 0.0154
    \end{bmatrix},
\end{equation*}
also shown in the middle plot of Figure~\ref{fig:2d_simulation}. There are a number of interesting features to note. First, the generalised update does not simply set the negative values to zero, as in \cite{swallow1984monte, wilkinson1995bayes, williamson2012fast}. For this $\mathcal{C}$, setting negative values to zero implies the norm in \eqref{eqn:objective} is with respect to Euclidean distance, rather than respecting the geometry of $\mathcal{X}_\bD$. Second to note, is that the generalised variance is not simply a rescaling of $\mathrm{var}_\bD[\bX]$; the principal axes of $\mathrm{var}_\bD[\bX]$ and $\mathrm{var}_\bd^\mathcal{C}[\bX]$ differ. The reasons for this can be inferred from the right plot of Figure~\ref{fig:2d_simulation}. Here we show $\mathcal{C}$ and the standard and generalised adjusted beliefs rotated onto the standardised space $\mathcal{Z}$ via $\bL^{-1}$ where $\bL = \bQ \sqrt{\blambda}$, as above, and so $\mathrm{Z}_1$ and $\mathrm{Z}_2$ are defined in terms of the principal axes of $\mathrm{var}_\bD[\bX]$. In the space $\mathcal{Z}$, the generalised update is with respect to the Euclidean norm, as $\mathrm{var}[\bZ] = \bI$. In this example we predominantly adjust along $\mathrm{Z}_1$, resulting in a large reduction of variance along this axis, and adjust a relatively smaller amount along $\mathrm{Z}_2$, with a comparatively smaller adjustment in variance along $\mathrm{Z}_2$. As $\bS \neq c \bI$, for some constant $c$, the eigenvectors of $\mathrm{var}_\bd^\mathcal{C}[\bX] = \bL \bS \bL^\T$ will differ from $\mathrm{var}_\bD[\bX] = \bL \bL^\T$.

\begin{figure}[h!]
    \centering
    \includegraphics[width = \linewidth]{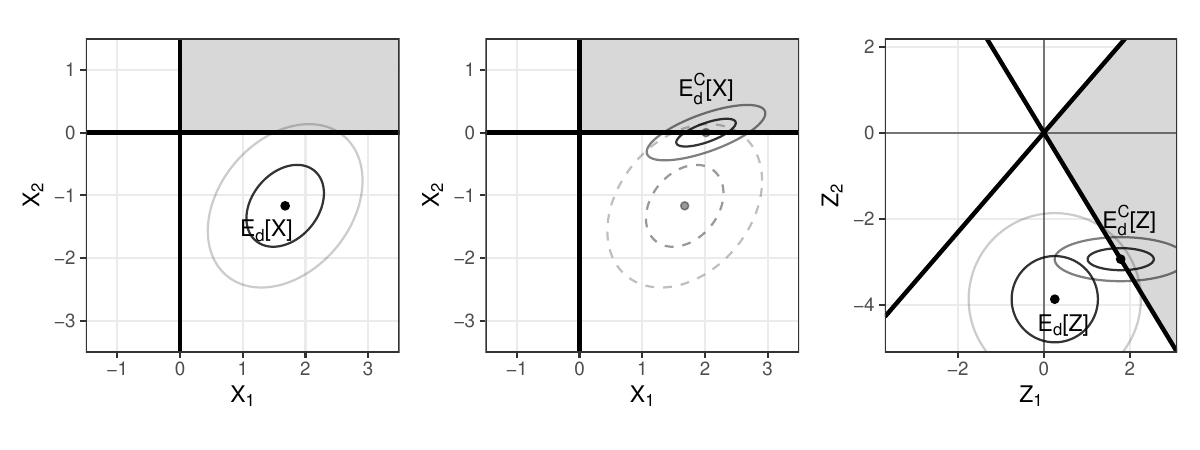}
    \caption{Left: the observed adjusted expectation $\mathbb{E}_\bd[\bX]$ and $\mathcal{C}$, the non-negative cone, shaded. Middle: the generalised adjusted expectation $\mathbb{E}_\bd^\mathcal{C}[\bX]$ as the projection of $\mathbb{E}_\bd[\bX]$ onto $\mathcal{C}$. Right: the rotation into the standardised space $\mathcal{Z}$; observed adjusted expectation and generalised adjusted expectation are represented by $\mathbb{E}_\bd[\bZ]$ and $\mathbb{E}_\bd^\mathcal{C}[\bZ]$. In all plots, corresponding uncertainties are represented by ellipses of one and two standard deviations.}
    \label{fig:2d_simulation}
\end{figure}

\section{Applications}\label{sec:applications}

The scope for application of this methodology are wide and varied. This is particularly so when $\mathcal{C}$ describes a closed convex set, and so a solution to \eqref{eqn:objective} may be guaranteed. We provide examples for two commonly encountered problems in statistical modelling: monotonic emulation and count regression. We conclude by providing some comments on other specifications for $\mathcal{C}$ that are expected to be of use. Before we continue, we reiterate that there are many bespoke existing solutions to each of the problem classes that we discuss; we review the relevant literature in context below. The novelty of our methodology is not that it provides any single methodological solution; but rather, it generalises to a very broad class of problems. Further, we note that this methodology is simple to implement and much quicker to compute than any of the considered simulation-based methodologies.

\subsection{Monotonic regression} \label{sec:monotone}

Consider a stochastic process, $\mathrm{X}(t)$, for $t \in \mathbb{R}^p$, and suppose we want to restrict $\mathrm{X}(t)$ to the space of monotone functions, $\mathcal{M}$, so that $\mathrm{X}(t) \leq \mathrm{X}(t')$ for $t \preceq t'$, and $t,t' \in \mathbb{R}^p$. Monotonic regression, particularly with Gaussian processes (GPs), has long received attention due to the need within many applied problems to incorporate physical knowledge within surrogates or emulators of computer models. In fact, monotonicity is just one important type of constraint we might want to build in, and our generalised Bayes linear methods provide the same elegant solution to all types of constraint. Existing solutions work reasonably well, but usually come at a significant computational cost. One approach, for instance taken by \cite{pepper2023probabilistic}, is to borrow the common projection onto a basis technique for dimension reduction from the Gaussian process literature \citep{higdon2008bayesian, williamson2012fast, salter2019uncertainty}, consider a basis of monotonic functions, project the data onto that and then fit surrogates to the coefficients. This can work well, however the class of functions is restricted by the basis and the coefficients from any basis projection are very frequently non-stationary and so not amenable to standard GP regression techniques. 

More principled approaches make use of the fixed sign of the derivative in a monotonic setting and assume differentiability of the true process. \cite{wang2016estimating} assume known derivatives and jointly update the GP with this information. Allowing for knowing the sign without knowing the derivatives and integrating them out is intractable and so either approximating distributions or constrained sequential Monte Carlo are used to generate samples \citep{riihimaki2010gaussian,golchi2015monotone}. Sampling methods come with significant computational overheads and there is no readily available software, meaning that practitioners with monotone problems face investing in laborious code development even to begin working with their data. \cite{lin2014bayesian} present a computationally expedient methodology that converts the GP into a monotonic stochastic process by projecting samples into the greatest convex minorant, where distance in the projection is measured with respect to the Euclidean norm. Though the resulting process is valid and monotonic, it is not clear that the resulting process formally respects any prior belief specification. 

To provide some more rigorous comparisons to existing methodology we repeat the simulation study in \cite{lin2014bayesian}, also studied in \cite{holmes2003generalized}, \cite{neelon2004bayesian}, and \cite{shively2009bayesian}. Data are generated from six mean functions, with Gaussian error and standard deviation $\sigma^2 = 1$. The functions are flat, $F_1(x) = 3$, sinusoidal, $F_2(x) = 0.32(x + \sin x)$, step, $F_3(x) = 3$ for $x \leq 8$ and $F_3(x) = 6$ for $x > 8$, linear, $F_4(x) = 0.3x$, exponential, $F_5(x) = 0.15 \exp (0.6x - 3)$, and logistic, $F_6(x) = 3/(1 + \exp(-2x + 10))$. They are evaluated at $n=100$ equidistant points on the interval $x \in [0,10]$, and we repeat each simulation 100 times. To repeat their study, the \cite{lin2014bayesian} method is applied to an initial GP, $f \sim \mathcal{GP}(\bm{0}, \bK)$, with $i$th $j$th elements of $\bK$ given by $\bK_{ij} = \eta \exp [-\gamma(x_i - x_j)^2]$, and $\{\sigma^2, \eta, \gamma\}$ estimated via their maximum a-posteriori values. For comparability we use the same parameters in our covariance function, $\mathrm{cov}[\mathrm{X}_i, \mathrm{X}_j]$, and keep a prior mean of $0$ for all $x$, hence the adjusted expectation and GP mean coincide. Note that a standard BL approach might specify the parameters directly or use cross-validation \citep{jackson2023efficient}. Root mean squared errors are shown in Table~\ref{tab:monotone} for the unconstrained GP, the generalised Bayes linear fit and the model of \cite{lin2014bayesian}; we label these models \textsc{gp}, \textsc{gbl} and \textsc{l\&d}, respectively. The only simulation where \textsc{l\&d} performs markedly differently is for the step function, the discontinuity of which is unfairly easy to estimate with their methodology as the underlying function is closer to the space of their projection rather than their initial beliefs. Table~\ref{tab:monotone} also offers a comparison of the computational speeds between the models. For \textsc{gp}, we report the timings for a rejection sampler to draw a monotonic sample. We set a maximum iteration count at 1000 and report, in the bottom row, the percent of simulations in which no samples were found. As \textsc{gp} and \textsc{l\&d} are sampling based methodologies, we report the times required for 1000 samples; in practice, the number of required samples may be much larger. In comparing the results we find that the generalised Bayes linear methodology not only offers a functional and generalisable framework, but is as accurate and faster than bespoke methodology.

\begin{table}[h!]
  \caption{Root mean squared errors (x100) for simulated processes as described in the main text. All computational timings are in centiseconds ($\mathrm{cs}$). Reported times for \textsc{gp} and \textsc{l\&d} are based on 1000 samples. The values recorded alongside \textsc{gp} denote the times required for a rejection sampler to sample from \textsc{gp}, note the scaling of $\times 10^3$ for all values. Here, we set a maximum iteration count to 1000; the percentage of simulations that do not yield a valid sample within 1000 samples is recorded in the bottom row. All results are averaged across 100 simulations, bracketed values denote standard deviations.}
  \small
  \begin{center}
      \begin{tabular}{llcccccc}
        & & Flat & Sinusoidal & Step & Linear & Exponential & Logistic \vspace{1mm} \\
        \hline
        \multirow{ 3}{*}{\rotatebox[origin=c]{90}{RMSE}} & \textsc{gp} & 17 (7.8) & 22 (6.2) & 46 (4.3) & 22 (7.7) & 30 (8.0) & 26 (7.1) \\
        & \textsc{gbl} & 12 (7.6) & 21 (6.0) & 45 (4.3) & 19 (6.8) & 23 (7.0) & 23 (7.4) \\
        & \textsc{l\&d} & 13 (7.7) & 20 (6.0) & 41 (5.4) & 20 (6.5) & 24 (6.5) & 22 (7.3) \vspace{1mm} \\
        \hline 
        \multirow{ 3}{*}{\rotatebox[origin=c]{90}{Time}} & \textsc{gp} ($\times 10^3$) & 36 (51) & 21 (38) & 52 (56) & 6.1 (20) & 35 (44) & 50 (54) \\
        & \textsc{gbl} & 7.3 (2.0) & 3.7 (2.2) & 9.3 (2.6) & 3.2 (1.7) & 6.3 (2.6) & 5.3 (2.2) \\
        & \textsc{l\&d} & 550 (200) & 78 (49) & 310 (90) & 42 (36) & 220 (71) & 150 (53) \vspace{1mm} \\ \hline
        & \% NA & 37 & 12 & 89 & 6 & 76 & 48 \vspace{1mm} \\
      \end{tabular}
  \end{center}
  \label{tab:monotone}
  \end{table}

\subsubsection{Emulating an afforestation uptake model}

We provide a more complex application to an agent based model for uptake of subsidies for tree planting in the UK. The simulator has $6$ inputs representing 5 levels of payment associated to planting new trees on private land, and a total budget devoted to the policy. The payments are per hectare for planting managed and unmanaged conifer woodlands ($\phi_1$, $\phi_2$) and deciduous woodlands ($\phi_3$ and $\phi_4$), and per tonne of greenhouse gases removed from the atmosphere $\phi_5$. The total budget is $\phi_6$ and, amongst many other things, the model outputs time series for various ecosystem services of interest to policymakers. For our application, we only consider the total area (in hectares) of trees planted, noting that there are targets in UK law \citep{EnvironmentAct21} on this number by 2050, as well as internal interim milestones. We use existing training data produced for the study described in \cite{nortier2024deep}, to emulate the area of planted trees as a function over the six dimensional parameter space, $\Phi = \{\phi_1, \dots, \phi_6\}$, and time, $t \in \{1, \dots, 150\}$. 

The simulations are conditionally (on $\Phi$) monotone so that the random process $\mathrm{X}(t ; \Phi)$, defined over time and parametrised by $\Phi$, is in $\mathcal{C} = \{\mathrm{X}(t ; \Phi) \in \mathbb{R} \mid \mathrm{X}(t ; \Phi) \leq \mathrm{X}(t' ; \Phi) \text{ for } t \leq t'\}$. Define $\bX(\Phi) \in \mathbb{R}^{150}$ as the $150$-dimensional multivariate output of all times at input $\Phi$. We train a traditional Bayes linear emulator to 366 model runs $\bD = (\bX(\Phi_1), \dots, \bX(\Phi_{366}))$ and withhold two model runs, corresponding to different policies, for prediction, $\bX = (\bX(\Phi_1^*), \bX(\Phi_2^*))$. We make the standard assumption of separability between the inputs and the outputs, so that for any two model runs, $\cov[\bX(\Phi_i), \bX(\Phi_j)] = \bK_t \otimes c(\Phi_i, \Phi_j)$ where $\bK_t \in \mathbb{R}^{150 \times 150}$ and $c(\Phi_i, \Phi_j)$ is a correlation function over the inputs. The $i,j$th elements of $\bK_t$ are defined via a Mat\'{e}rn-$5/2$ covariance function, $c_{5/2}(\cdot)$, with a length-scale of $10$ years and amplitude $\sigma = 1.8\times 10^7$: ${\bK_t}_{[i,j]} = \sigma^2 c_{5/2}(i, j; 10)$. Further, $c(\Phi_i, \Phi_j)$ is defined as $c(\Phi_i, \Phi_j) = \prod_{k=1}^6 c_{5/2}(\phi_{k,i}, \phi_{k,j}; l_k)$ with length scales $\{l_1, \dots, l_6\} = \{3.0, 1.4, 1.3, 1.6, 0.17, 1.0\}$, where the $\phi_{k,j}$ represents parameter $k$ at location $j$. The emulator results for the two withheld policies are shown in the left images of Figure~\ref{fig:trees}. The black lines and shaded regions are the emulator expectation and $\pm 2$ marginal standard deviations, respectively. Marginal standard deviations are also represented by the dashed line and correspond to the right hand y-axis. The dots are the actual simulator values, and we stress that these points are not in the training set. The right images of Figure~\ref{fig:trees} show a generalised Bayes linear emulator trained with respect to the monotonicity constraints in $\mathcal{C}$. There are some interesting points to note. Policy 1 is further from $\mathcal{C}$ than Policy 2. Accordingly, the generalised Bayes linear emulator in Policy 1 has a larger reduction in variance than we see in Policy 2, which only requires a minor adjustment, obeying the constraints that we detail in Section~\ref{sec:gen_adj_var}. The root-mean-squared-errors for the unconstrained emulator are $2.86 \times 10^6 \, \mathrm{ha}$ and $2.10 \times 10^6 \, \mathrm{ha}$ and for the monotone emulator are $1.72 \times 10^6 \, \mathrm{ha}$ and $1.90 \times 10^6 \, \mathrm{ha}$, and so in this instance, enforcing our solution to lie in $\mathcal{C}$ results in a more accurate emulator in mean-square-error. The additional computation incurred by the generalised Bayes linear update is negligible: $\sim 0.5 \, \mathrm{sec}$ on a standard workstation, with 32GB RAM and an Apple M2 Pro processor.

\begin{figure}
    \centering
    \includegraphics[width = \linewidth]{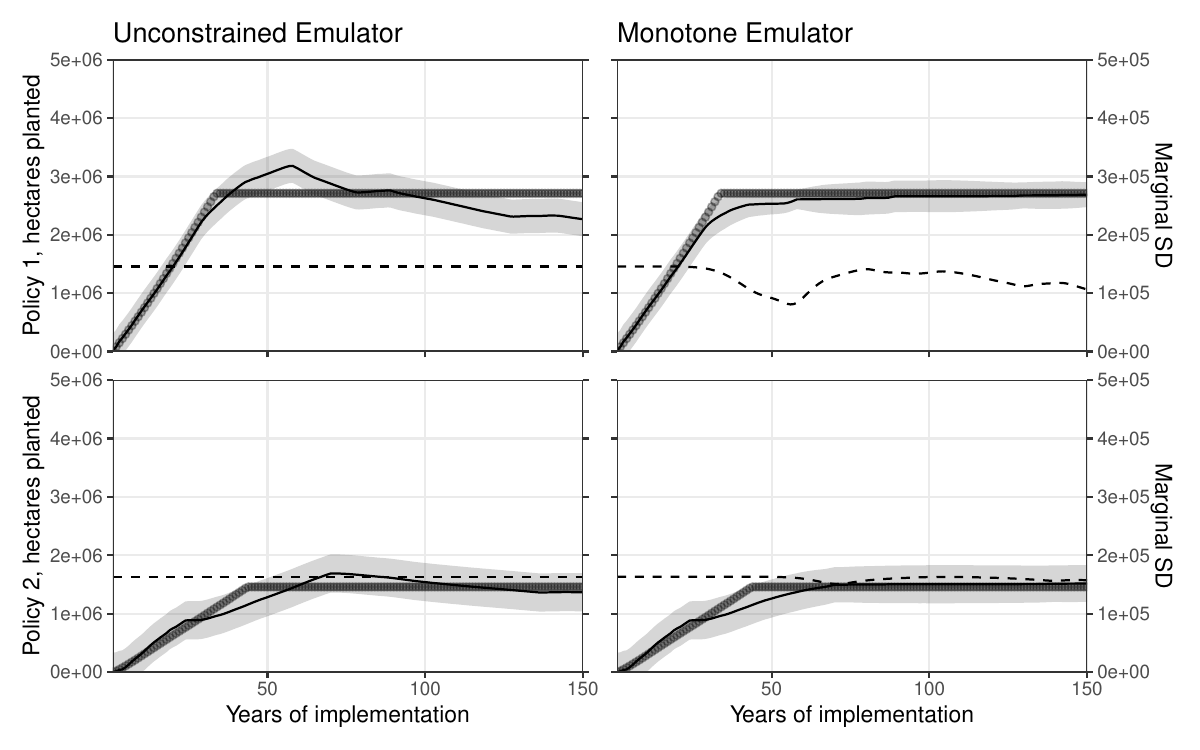}
    \caption{Emulators of hectares of trees planted under two policies and over 150 years, trained on 366 simulation runs. Left images show the emulator predictions from an unconstrained Bayes linear emulator; Right images show results for a generalised Bayes linear emulator trained with monotonicity constraints. The solid lines are the predicted expectation, and shaded regions denote $\pm 2$ standard deviations, also represented by the dashed lines and right hand axis.}
    \label{fig:trees}
\end{figure}

\subsection{Count processes}

As noted in Remark~\ref{rem:bl}, Bayes linear statistics may be used to infer beliefs of varied data, and not just that which exhibits Gaussian like behaviour. Here we provide an example with count data of daily deaths by Lower Tier Local Authority (LTLA) during the COVID19 pandemic in the UK. Spatio-temporal modelling was important during the pandemic, giving government and the public a sense of the state of the disease, through the exponential growth rate (r-number), and the number of deaths being caused. Inferring the growth rate using reported cases or inferring number of deaths given partially reported deaths could be done using Bayesian hierarchical models \citep{guzman2023bayesian}, having a counts model with a rate parameter representing the expected count, and a positive transformation of that parameter following some spatio-temporal process. Whilst a Bayes linear analysis might be performed over, say, logarithmic-deaths, to ensure positivity of the policy relevant random quantity, such an analysis requires beliefs to be held by the analyst in the transformed space and only allows reporting on the policy relevant quantity via bounds derived through Jensen's inequality. Making the quantity of interest number of deaths so that meaningful prior judgements might be understood or even set by health security agencies or policy makers, can often lead to negative estimates at times or in locations when there are few deaths. 

To illustrate and to demonstrate the efficacy of generalised Bayes linear, we consider the number of deaths by LTLA on the 17th of June 2020 (with COVID19 listed on the death certificate). The data are shown in the top left plot of Figure~\ref{fig:covid}. We specify a relatively simple covariance structure, assuming a squared exponential covariance function over space, with the position of each LTLA given by its geometrical centre. We define by $\bX = (\mathrm{X}_1, \dots, \mathrm{X}_{376})$ the unknown mean daily COVID19 deaths for each of the $376$ LTLAs at known geometrical centres $(s_1, \dots, s_{376})$, and by $\bD = (\mathrm{D}_1, \dots, \mathrm{D}_{354})$ the $354$ reported death counts at geometrical centres $(r_1, \dots, r_{354})$. The data are obtained from \cite{UKHSA2024} and are missing values for all Wales LTLAs and Buckinghamshire; we treat this data as missing at random as its missingness is due to clerical reasons. We specify the beliefs $\mathbb{E}[\mathrm{X}_i] = \mathbb{E}[\mathrm{D}_j] = 0$ for $i \in \{1, \dots, 376\}$ and $j \in \{1, \dots, 354\}$, and
\begin{align*}
    \mathrm{cov}[\mathrm{X}_i, \mathrm{X}_j] = \exp&\left(-l^{-2}\norm{s_i - s_j}^2\right), \quad \mathrm{cov}[\mathrm{X}_i, \mathrm{D}_j] = \exp\left(-l^{-2}\norm{s_i - r_j}^2\right) \\
    \mathrm{cov}[\mathrm{D}_i, \mathrm{D}_j] &= \exp\left(-l^{-2}\norm{r_i - r_j}^2\right) + 0.5^2 \mathds{1}\{r_i = r_j\},
\end{align*}
for $l = 85 \mathrm{km}$ where norms, here, are taken to represent geodesic distance. Note that the choice $0.5$ implies that the reported counts being out by 1 is a 2 standard deviation event. We might want to increase this in situations were we might believe COVID19 deaths are misreported. 

 Given the observed data for 17th of June 2020, the adjusted expectation of deaths, $\mathbb{E}_\bd[\bX]$, is shown in the bottom left of Figure~\ref{fig:covid}. Regions $i$ with $\mathbb{E}_\bd[\mathrm{X}_i] < 0$ are shaded and outlined in red. As in Section~\ref{sec:example}, we set $\mathcal{C} = \{\bX \in \mathbb{R}^{354} \mid \bX \succeq 0\}$ to be the non-negative cone and calculate the generalised Bayes linear update, $\mathbb{E}_\bd^\mathcal{C}[\bX]$. This is shown in the bottom right of Figure~\ref{fig:covid}. Due to colour saturation, differences between these updates may not be apparent; we plot these differences in the top right of Figure~\ref{fig:covid}. As opposed to simply zeroing-out the negative valued regions we can see changes to many adjacent LTLAs such that a solution is provided that is most coherent with the initial belief specification. Increasing the lack of confidence in the reporting mechanism can exacerbate the negative estimates in the standard case, making the need for the generalised method all the more acute.

\begin{figure}
    \centering
    \includegraphics[width = 0.8\linewidth]{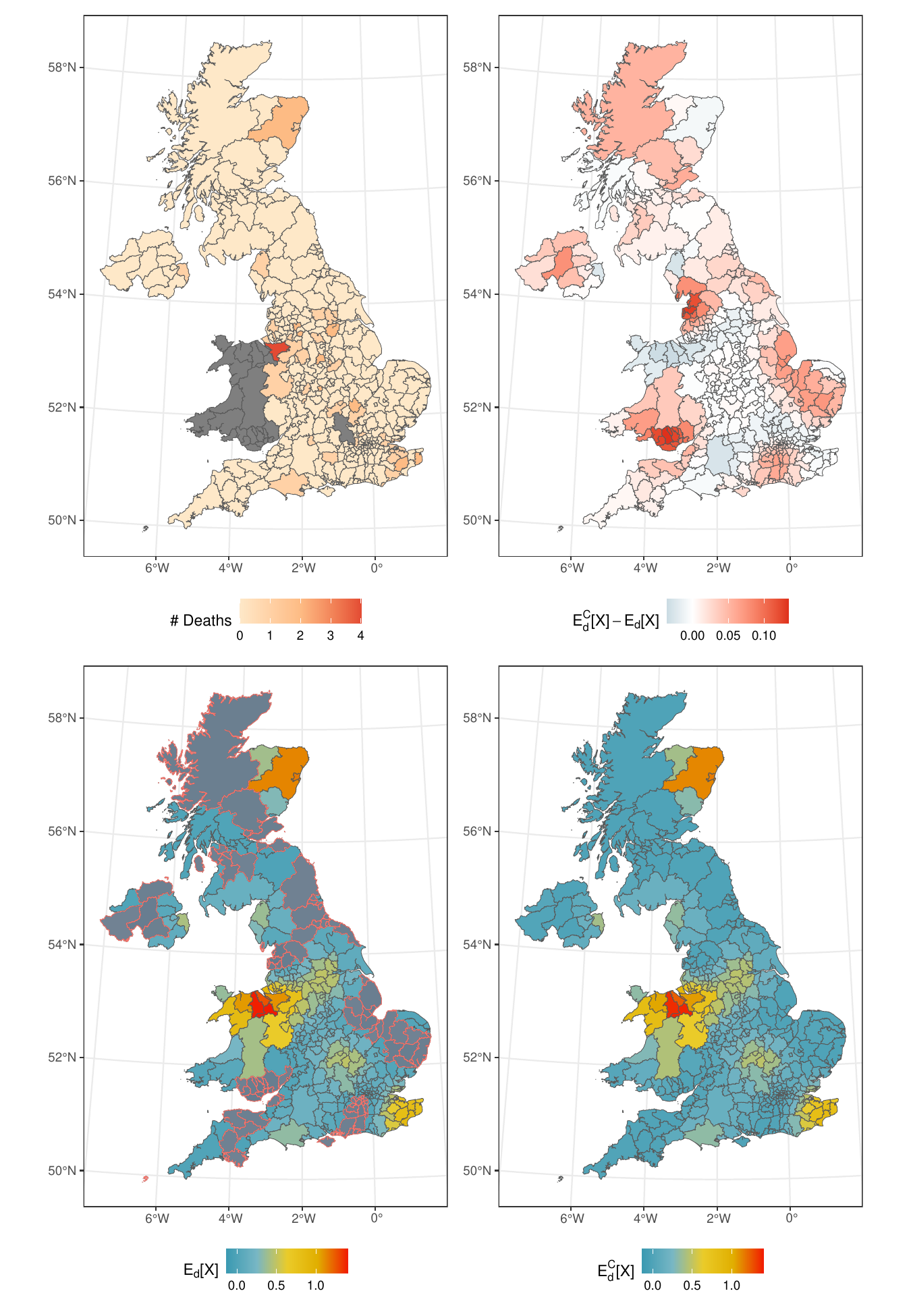}
    \caption{Deaths by Lower Tier Local Authority with COVID19 listed on death certificate on the 17th of June 2020 and statistical fits. The top left shows number of deaths, missing data are shaded grey. The bottom left shows the standard Bayes linear fit to the data, subject to the belief specifications provided in the main text. Regions with negative adjusted expectation are shaded. The bottom right shows the generalised Bayes linear adjustment and the top right shows the difference between the standard and generalised adjustments.}
    \label{fig:covid}
\end{figure}

It is worth noting that, even if a hierarchical Bayesian model is desirable to an analyst in this particular application (or at higher resolution with many more geometrical centres), such models take a long time to fit, perhaps needing the type of sparse precision matrix obtained through bespoke packages \citep{lindgren2015bayesian}. The generalised Bayes linear framework offers a fast way to obtain a prior analysis that might be used under time pressure or used to guide model selection. 

\subsection{Further examples}

We have demonstrated the potential for generalised Bayes linear updating for monotonic and non-negative solution spaces. There are a number of further specifications for $\mathcal{C}$ that can also be useful for common problems in statistical modelling; we focus on examples with convex $\mathcal{C}$ so that a solution is guaranteed. We may estimate probabilities in logistic regression by bounding the values in $\mathcal{C}$ below by 0 and above by 1. Where probability is defined as a statistical process over a domain, we may quickly and easily update our beliefs on the statistical process directly, rather than employing sampling based estimation procedures such as P{\'o}lya--Gamma data augmentation \citep[see][]{polson2013bayesian}. Should $\bX$ instead be a random quantity defined over positive semi-definite matrices, as in \cite{wilkinson1995bayes}, $\mathcal{C}$ may describe the positive semi-definite cone $\mathcal{C} = \{\bX \in \mathbb{R}^{n \times n} \mid x^\T \bX x \geq 0, \; \text{for} \; x \in \mathbb{R}^n\}$. When our adjusted beliefs are not positive semi-definite it is typical to calculate the closest matrix subject to the Frobenius norm; the matrix equivalent to the Euclidean norm \citep[e.g. in][]{astfalck2024coexchangeable, gandy2023tuning}. Further, we may consider adjusting beliefs for $\bX$ whose derivatives are subject to linear constraints. We have already seen an example of this in Section~\ref{sec:monotone} where we may alternatively define $\mathcal{C}$ as $\mathcal{C} = \{\mathrm{X}(t) \in \mathbb{R} \mid \frac{\mathrm{d}}{\mathrm{d}t} \mathrm{X}(t) \geq 0\}$. We may readily extend this concept by restricting $\mathrm{d} \mathrm{X}(t)/\mathrm{d}t$ to other convex subsets, perhaps including an upper bound to limit the rate of growth of $\mathrm{X}(t)$; or to higher order derivatives, for instance setting $\mathrm{d}^2 \mathrm{X}(t)/\mathrm{d}t^2 \geq 0$ or $\mathrm{d}^2 \mathrm{X}(t)/\mathrm{d}t^2 \leq 0$ to enforce $\mathrm{X}(t)$ to be convex or concave, respectively. 

\section{Discussion}\label{sec:discuss}

Bayesian inference as optimisation continues to grow in popularity. We have made the case for a framing of generalised Bayesian inference as one in which an optimal point within a restricted solution space endowed with a geometry is defined to be that which is closest to a generalised data generating process and where the notion of close is specified by the practitioner. Current GBI methods operate in $\mathcal{L}_{2}(\Theta,\mu)$, the space of functions $\theta$ with measure $\mu$ with integral inner product, restrict the solution space to probability density functions (or even certain classes of tractable densities), and tend to use different types of divergence to define the notion of `close'. Bayes linear methods work in Hilbert spaces with product inner product (which naturally define the notion of closeness) and search the solution space of affine functions of the data. 

Inspired by this reformulation, we proposed a principled generalisation to Bayes linear methods to enable restricted solution spaces as subsets of the original Hilbert space. The result is an extremely fast class of Bayes as optimisation methods that can be used under partial prior specification to make inference, or can be used as part of exploratory modelling and data analysis to inform and shape novel statistical modelling or computationally burdensome analyses. Our applications to monotonic regression and spatial counts are the first Bayes linear solutions to these problems, and demonstrate that we have dramatically widened the class of problems that can be addressed within the theory. 

We hope that, by pointing out the overarching geometric framework covering generalised methods (including ours), the GBI community or others might be as inspired to revisit the properties of their methods as we were with ours. In our case, the restriction of the solution space to affine functions, but without necessary coherence restrictions was revisited to provide new solutions. The GBI community are already actively examining alternative divergences and restrictions on solution spaces (Properties 2 and 3). Perhaps further generalisations, for example to the underlying geometries and inner products will open different avenues of research. 

\section*{Acknowledgements}

\if1\blind
{LA is supported by the ARC ITRH for Transforming energy Infrastructure through Digital Engineering (TIDE), Grant No. IH200100009. Danny Williamson was funded by EPSRC projects UQ4Covid (EP/V051555/1) and ADD-TREES (EP/Y005597/1). Cassandra Bird was funded by EPSRC studentship EP/V520317/1. The authors would like to thank Bertrand Nortier and Mattia Mancini for sharing their existing ensemble of tree uptake simulations.} \fi
\if0\blind
{ \textit{Acknowledgements have been redacted for blind review.}} \fi

\bibliographystyle{model2-names.bst}
\bibliography{references}

\appendix

\section{Derivations of Bayes linear updating equations}

\subsection{Proof of adjusted expectation}

Adjusted expectation $\mathbb{E}_\bD[\bX]$ is defined by the orthogonal projection of $\bX$ onto $\mathcal{D} \cup \textcal{1}$. We thus have $\mathbb{E}_\bD[\bX] = \bh_0 + \bH_0 \bD$ where $\bh_0$ and $\bH_0$ ensure
\begin{equation}
    \langle \bX - \bh_0 - \bH_0 \bD, \bh + \bH \bD \rangle = \mathbb{E}[(\bX - \bh_0 - \bH_0 \bD)^\intercal (\bh + \bH \bD)] = 0
\end{equation}
for all $\bh \in \mathbb{R}^{n \times 1}$ and $\bH \in \mathbb{R}^{n \times m}$. Spaces $\mathcal{D}$ and $\textcal{1}$ are orthogonal and so we may equivalently solve for $\bh_0$ and $\bH_0$ so that
\begin{align}
    \mathbb{E}[(\bX - \bh_0 - \bH_0 \bD)^\intercal \bh] &= \sum_i \bh_i \mathbb{E}[(\bX - \bh_0 - \bH_0 \bD)_i] = 0 \text{, and} \\
    \mathbb{E}[(\bX - \bh_0 - \bH_0 \bD)^\intercal \bH \bD] &= \sum_i \sum_j \bH_{ij} \mathbb{E}[(\bX - \bh_0 - \bH_0 \bD)_i \bD_j] = 0.
\end{align}
As this solution must hold for all $\bh \in \mathbb{R}^{n \times 1}$ and $\bH \in \mathbb{R}^{n \times m}$, we find that
\begin{align}
    \mathbb{E}[\bX - \bh_0 - \bH_0 \bD] &= \zero^{n \times 1} \text{, and} \label{eqn:adj_exp_a1}\\ 
    \mathbb{E}[(\bX - \bh_0 - \bH_0 \bD) \bD] &= \zero^{n \times m}. \label{eqn:adj_exp_a2}
\end{align}
Thus, from (\ref{eqn:adj_exp_a1}), we find that $\bh_0 = \mathbb{E}[\bX] - \bH_0 \mathbb{E}[\bD]$ and from (\ref{eqn:adj_exp_a2}) we find $\bH_0 = \cov[\bX, \bD] \var[\bD]^\dagger$. The adjusted expectation is thus solved as
\begin{equation}
    \mathbb{E}_\bD[\bX] = \bh_0 + \bH_0 \bD = \mathbb{E}[\bX] + \cov[\bX,\bD] \var[\bD]^\dagger(\bD - \mathbb{E}[\bD]).
\end{equation}

\subsection{Proof of adjusted variance} \label{app:adj_var}

The adjusted variance is a description of the remaining uncertainty in $\bX$ after $\bD$ is observed. It is defined by the outer product
\begin{equation} \label{eqn:adj_var_a1}
    \var_\bD[\bX] = \mathbb{E}[(\bX - \mathbb{E}_\bD[\bX])(\bX - \mathbb{E}_\bD[\bX])^\intercal]
\end{equation}
and characterises the totality of norm expressions of the linear subspace $\bX - \mathbb{E}_\bD[\bX]$. Substituting $\mathbb{E}_\bD[\bX] = \mathbb{E}[\bX] + \cov[\bX,\bD] \var[\bD]^\dagger(\bD - \mathbb{E}[\bD])$ into (\ref{eqn:adj_var_a1}) yields
\begin{align}
    \var_\bD[\bX] &= \mathbb{E}[(\tilde{\bX} - \cov[\bX,\bD] \var[\bD]^\dagger \tilde{\bD})(\tilde{\bX} - \cov[\bX,\bD] \var[\bD]^\dagger \tilde{\bD})^\intercal] \\
    &= \var[\bX] - \cov[\bX, \bD] \var[\bD]^\dagger \cov[\bD,\bX]
\end{align}
where, here, $\tilde{\bX} = \bX - \mathbb{E}[\bX]$ and $\tilde{\bD} = \bD - \mathbb{E}[\bD]$.

\section{Proofs of Section~2}

% This is a proof for a probability model that generates the belief structure. 

\subsection{Proof of Theorem~\ref{the:bl}}

This proof demonstrates the assertion that many common probabilistic Bayesian posterior distributions share the some analytical form as the Bayes linear updating equations. As such, we start from a probability measure that generates $\E[\bD]$, $\E[\bX]$, $\var[\bD]$, $\var[\bX]$, and $\cov[\bD, \bX]$, and we align notation with standard notions of condition expectation in probability theory, rather than the Bayes linear adjusted expectations. We extend the proof in \cite{ericson1969note} to show that the form for the adjusted expectation in \eqref{eqn:adj_exp} is identical to the posterior expectation when we know the posterior expectation is linear in $\bD$, $\E[\bX \mid \bD] = \bA \bD + \bB$. Similar results are also found in \cite{hartigan1969linear}. First, using the law of iterated expectation
\begin{equation*}
    \E[\bX] = \E_{\bD}\left[\E_{\bX}[\bX \mid \bD]\right] = \bA \E[\bD] + \bB
\end{equation*}
and so $\E[\bX] - \bA \E[\bD] = \bB$. We also have, by definition of covariance, $\E[\bD \bX^\T] = \cov[\bD, \bX] + \E[\bD] \E[\bX]^\T$ which we equate with
\begin{align*}
    \E[\bD \bX^\T] &= \E_\bD[\bD \E_\bX[\bX \mid \bD]^\T] = \E_\bD[\bD (\bA \bD + \bB)^\T] \\
    &= \var[\bD]\bA^\T + \E[\bD]\E[\bD]^\T\bA^\T + \E[\bD]\E[\bX]^\T - \E[\bD]\E[\bD]^\T\bA^\T \\
    &= \var[\bD]\bA^\T + \E[\bD]\E[\bX]^\T
\end{align*}
to yield the solution $\bA = \cov[\bD, \bX]\var[\bD]^\dagger$. Substituting this, along with $\bB = \E[\bX] - \bA \E[\bD]$, into $\E[\bX \mid \bD]$ yields \eqref{eqn:adj_exp}. Given \eqref{eqn:adj_exp}, the posterior variance $\var[\bX \mid \bD]$ follows \eqref{eqn:adj_var}, proof of this follows similarly to that provided in Appendix~\ref{app:adj_var}.

\subsection{Further comments to Remark~\ref{rem:bl}}

An important family of probability distributions that exhibits posterior linearity in $\bD$ is data generated from an exponential family with conjugate prior distribution. To state this more formally, assume $\bD \sim p(\bD \mid \bX)$ where $p(\bD \mid \bX) \propto \exp\{\eta(\bX) \cdot \bD - \kappa(\bX)\}$ is the exponential family in canonical form, and assume the prior distribution $p(\bX) \propto \exp\{n_0 x_0 \cdot \bX - n_0 \kappa(\bX)\}$. First shown in \cite{diaconis1979conjugate} for the general case of multivariate $\bX$, it is a well known result that the posterior expectation under this model is linear in $\bD$ and so the results above hold. The implications of which are that the Bayes linear updating equations in \eqref{eqn:adj_exp} and \eqref{eqn:adj_var} also describe the first two posterior moments for many common probabilistic Bayesian models, for instance Poisson/Gamma, Bernoulli/Beta and Gamma/Gamma combinations of likelihoods/priors, and not just the commonly considered case of Gaussians.

\end{document}